\documentclass[letterpaper, 10 pt, conference]{ieeeconf} 
\IEEEoverridecommandlockouts
\usepackage[utf8]{inputenc}
\usepackage{epsfig}
\usepackage{caption}
\usepackage{subcaption}
\usepackage{graphicx}
\usepackage{amsfonts}
\usepackage{amsmath}
\usepackage{amssymb}
\usepackage{color}
\usepackage{pstricks}
\usepackage{url}
\usepackage{cite}
\usepackage{longtable}
\usepackage{algorithm}
\usepackage{algorithmic}
\newtheorem{theorem}{Theorem}

\usepackage{tikz}

\usepackage[normalem]{ulem}

\title{Stochastic Deep Model Reference Adaptive Control}

\author{Girish Joshi and Girish Chowdhary
	\thanks{*Supported by the Laboratory Directed Research and Development program at Sandia National Laboratories, a multi-mission laboratory managed and operated by National Technology and Engineering Solutions of Sandia, LLC., a wholly owned subsidiary of Honeywell International, Inc., for the U.S. Department of Energy's National Nuclear Security Administration under contract DE-NA-0003525.}
	\thanks{Authors are with Coordinated Science Laboratory, University of Illinois,
		Urbana-Champaign, IL, USA
		{\tt\small girishj2@illinois.edu,girishc@illinois.edu}}%
}

\DeclareUnicodeCharacter{2212}{-} 
\begin{document}
	
\maketitle \thispagestyle{empty} \pagestyle{empty}

\begin{abstract}

In this paper, we present a Stochastic Deep Neural Network-based Model Reference Adaptive Control. Building on our work ``Deep Model Reference Adaptive Control", we extend the controller capability by using Bayesian deep neural networks (DNN) to represent uncertainties and model non-linearities. Stochastic Deep Model Reference Adaptive Control uses a Lyapunov-based method to adapt the output-layer weights of the DNN model in real-time, while a data-driven supervised learning algorithm is used to update the inner-layers parameters. This asynchronous network update ensures boundedness and guaranteed tracking performance with a learning-based real-time feedback controller. A Bayesian approach to DNN learning helped avoid over-fitting the data and provide confidence intervals over the predictions. The controller's stochastic nature also ensured "Induced Persistency of excitation," leading to convergence of the overall system signal.

\end{abstract}

\section{Introduction}
Deep Model Reference Adaptive Control (DMRAC) is a control architecture that seeks to learn a high-performance control policy in the presence of significant model uncertainties while guaranteeing stability \cite{joshi2019deep, joshi2020asynchronous,joshi2020design}. DMRAC leverages deep neural networks, which have demonstrated their learning power in many supervised learning applications in computer vision, natural language processing, etc.\cite{goodfellow2016deep}. However, using deep networks for learning control policies can lead to system instability, as is well known in the adaptive control literature \cite{annaswamy_CDC_89, dydek2010adaptive}. Furthermore, learning the parameters of a multi-layer neural network can be computationally intensive and hence cannot be updated in real-time. DMRAC addresses both issues through an asynchronous weight update law for a deep neural network model of the uncertainty. DMRAC temporally separates the slower learning of the inner layers, which form features for the fast learning output layer weights. The outer-most layer weights are adapted in real-time using Lyapunov-based Adaptive laws that have known stability guarantees \cite{annaswamy_CDC_89, Ioannou:96bk,Pomet_92_TAC, Chowdhary:CDC:10}. This online learning of the last layer ensures minimization of system tracking error and provides the necessary robustness and guaranteed tracking performance. The slower learning of the inner-layer weights of the DNN can be done over batches of data, either online or offline, over judiciously collected data sets. The inner layer features can thus be optimized for the final layer, which otherwise can be very difficult to be designed \cite{liu2018gaussian}.  Thereby, the DMRAC controller circumvents the entire feature design problem for complex nonlinear control problems, much like in other deep neural network applications. 

DMRAC takes one critical step towards verifiable deep learning-based adaptive control and improves the power of the learning models for adaptive control. In \cite{joshi2020asynchronous} DMRAC was demonstrated to provide stable controller in the presence of severe faults on a real quadrotor and showed to have a higher level of performance than shallower methods such as Gaussian Process-based MRAC (GP-MRAC) \cite{joshi2018adaptive, chowdhary:2013:TNNLS}. The asynchronous framework used in DMRAC also shows much promise for safety-critical applications and for other learning-based control methods, such as Deep Reinforcement Learning (D-RL). 
However, DMRAC or traditional DNN based learning methods cannot provide confidence bound on their predictions, and are in general susceptible to the critical issue of overfitting, leading to loss of generalizability which could be detrimental to long-term learning. Also, it is well known that when applied to supervised learning problems, deep networks are often incapable of correctly assessing the uncertainty in the training data and often make overly confident decisions about the correct class or prediction.

To address this issue of over-fitting of deep neural networks in estimating model uncertainty in DMRAC, we introduce in this paper a Stochastic-DMRAC (S-DMRAC) architecture using Bayesian neural networks. Apart from the benefits of Bayesian methods in adding a natural regularization to feature learning, we show that using finitely exciting training data, a desirable property of "\textit{Induced Persistency of Excitation}" of the features can be obtained. The persistency of excitation property of the deep features enables the outer layer weights to converge to their ideal values and therefore ensuring a desirable property of convergence of all signals. We empirically demonstrate the controller performance on the wing-rock system and compare the results against the DMRAC controller. We show both empirically and theoretically, the S-DMRAC network weights are well behaved and demonstrate boundedness given the induced persistency of excitation property of stochastic deep features.
\section{Background} 
\subsection{Deep Networks and Feature spaces in machine learning}\label{sec:feature} 
The key idea in machine learning is that a given function can be encoded with weighted combinations of  \textit{feature} vector $\Phi \in \mathcal{F}$, s.t $\Phi(x)=[\phi_1(x), \phi_2(x),...,\phi_k(x)]^T\in \mathbb{R}^k$, and $W^*\in\mathbb{R}^{k \times m}$ a vector of `ideal' weights s.t $\|y(x)-W^{*^T}\Phi(x)\|_\infty<\epsilon(x)$.
Instead of hand picking features, or relying on polynomials, Fourier basis functions, comparison-type features used in support vector machines \cite{schoelkofp:01,scholkopf2002learning} or Gaussian Processes \cite{rasmussen2006gaussian}, DNNs utilize composite functions of features arranged in a directed acyclic graphs, i.e. $\Phi(x)=\phi_n(\theta_{n-1},\phi_{n-1}(\theta_{n-2},\phi_{n-2}(...))))$ 
where $\theta_i$'s are the layer weights. The universal approximation property of the DNN with commonly used feature functions such as sigmoidal, tanh, and RELU is proved in the work by Hornik's \cite{Hornik:NN89} and shown empirically to be true by recent results \cite{2016arXiv160300988M,Poggio2017,2016arXiv161103530Z}. 
Hornik et al. argued the network with at least one hidden layer (also called Single Hidden Layer (SHL) network) to be a universal approximator. However, empirical results show that the networks with more hidden layers show better generalization capability in approximating complex function. While the theoretical reasons behind better generalization ability of DNN are still being investigated \cite{2016Matus,liang2016deep}, for the purpose of this paper, we will assume that it is indeed true, and focus our efforts on designing a practical and stable control scheme using DNNs.

\section{System Description}
\label{system_description}
This section discusses the formulation of Deep Model Reference Adaptive Control \cite{joshi2019deep}. We consider the following system with  uncertainty $\Delta(x)$:
\begin{equation}
\dot x(t) = Ax(t) + B(u(t) + \Delta(x)).
\label{eq:0}
\end{equation}
where $x(t) \in \mathbb{R}^n$, $t \geqslant 0$ is the state vector, $u(t) \in \mathbb{R}^m$, $t \geqslant 0$ is the control input, $A \in \mathbb{R}^{n \times n}$, $B \in \mathbb{R}^{n \times m}$ are known system matrices and we assume the pair $(A,B)$ is controllable. The term $\Delta(x) : \mathbb{R}^n \to \mathbb{R}^m$ is matched system uncertainty and be Lipschitz continuous in $x(t) \in \mathcal{D}_x$. Let $ \mathcal{D}_x \subset \mathbb{R}^n$ be a compact set and the control $u(t)$ is assumed to belong to a set of admissible control inputs of measurable and bounded functions, ensuring the existence and uniqueness of the solution to \eqref{eq:0}.

The reference model is assumed to be linear and therefore the desired transient and steady-state performance is defined by selecting the system eigenvalues in the negative half plane. The desired closed-loop response of the reference system is given by
\begin{equation}
\dot x_{m}(t) = A_{m}x_{m}(t) + B_{m}r(t).
\label{eq:ref model}
\end{equation}
where $x_{m}(t) \in \mathcal{D}_x  \subset \mathbb{R}^{n}$ and $A_{m} \in \mathbb{R}^{n \times n}$ is Hurwitz and $B_{m} \in \mathbb{R}^{n \times r}$. Furthermore, the command $r(t) \in \mathbb{R}^{r}$ denotes a bounded, piece wise continuous, reference signal and we assume the reference model (\ref{eq:ref model}) is bounded input-bounded output (BIBO) stable \cite{ioannou1988theory}.

\subsection{Deep MRAC: Total Controller}
\label{adaptive_identification}
The aim is to construct a feedback law $u(t)$, $t \geqslant 0$, such that the state of the uncertain dynamical system (\ref{eq:0}) asymptotically tracks the state of the reference model (\ref{eq:ref model}) despite the presence of matched uncertainty.

A tracking control law consisting of linear feedback term $u_{pd} = k_xx(t)$, a linear feed-forward term $u_{crm} = k_rr(t)$ and an adaptive term $u_{ad}(t)$ form the total controller 
\begin{equation}
u = u_{pd} + u_{crm} - u_{ad}.
\label{eq:total_Controller}
\end{equation}
The baseline full state feedback and feed-forward controller is designed to satisfy the matching conditions such that $A_{m} = A-Bk_x$ and $B_{m} = Bk_r$. For the adaptive controller ideally we want $u_{ad}(t) = \Delta(x(t))$.

The true uncertainty $\Delta(x)$ in unknown, but it is assumed to be continuous over a compact domain $\mathcal{D}_x \subset \mathbb{R}^n$. A Deep Neural Networks (DNN) have been widely used to represent a function when the basis vector is not known.  Using DNNs, a non linearly parameterized network estimate of the uncertainty can be written as $f_\theta(x) \triangleq \theta_n^T\Phi(x)$,  
where $\theta_n \in \mathbb{R}^{k \times m}$ are network weights for the final layer and $\Phi(x)=\phi_n(\theta_{n-1},\phi_{n-1}(\theta_{n-2},\phi_{n-2}(...))))$, is a $k$ dimensional feature vector which is function of inner layer weights, activations and inputs. The basis vector $\Phi(x) \in \mathcal{F}: \mathbb{R}^{n} \to \mathbb{R}^{k}$ is considered to be Lipschitz continuous to ensure the existence and uniqueness of the solution (\ref{eq:0}). Therefore the adaptive part of the controller can be written as
\begin{equation}
    u_{ad}(t) = W^T\Phi(x(t)).
    \label{eq:adaptive_Controller}
\end{equation}
Where the network weights $\boldsymbol\theta$ are updated minimizing a loss function $\ell(\boldsymbol{Z_i, \theta})$ over a batch of carefully collected locally exciting data $(\boldsymbol{Z_i})$ using Model Reference Generative Network (MRGeN) \cite{joshi2018adaptive, joshi2019deep}. 
$$\boldsymbol{\theta} = \arg \max_\theta \left(-\frac{1}{M}\sum_{i=1}^M \ell(\boldsymbol{Z_i, \theta})\right).$$

And $W$ are the fast updated weights, which are learned according to MRAC rule.
$$\dot{W} = -\Gamma\Phi(x(t))e(t)^TpB.$$
The weight $W$ will replace the final layer network weights $\theta_n$ in $f_\theta$ to form the adaptive control element \eqref{eq:adaptive_Controller} in the total controller \eqref{eq:total_Controller}.

Developing on this idea of a Lyapunov-based  method for adaptation laws of the output-layer weights of a DNN model in real-time while a data-driven supervised learning algorithm is used to update the inner-layer weights of the DNN; we present in this paper a stochastic deep MRAC architecture using Bayesian learning (S-DMRAC).

\section{Bayesian Deep Neural Networks}
Bayesian Neural Networks(BNN) are comprised of a Probabilistic Model and a Neural Network. The intent of such a design is to
combine the strengths of Neural Networks and
Stochastic modeling. Neural Networks exhibit
universal approximation property over continuous functions \cite{hornik1991approximation}. Probabilistic models  allows us to combine prior and likelihood to
produce probabilistic guarantees on the model predictions through generating distribution over the parameters learned from the observations. Therefore the network can specify the confidence interval over the prediction. Thus BNNs are a unique combination of neural network and stochastic models. 

The layer weights in the Bayesian neural networks are represented by probability distributions over possible values, rather than a single fixed value as in the traditional neural networks. The trained deep features and therefore adaptive element will be a stochastic quantity. 
The amount of perturbation each weight exhibits is learnt in a way that coherently explains variability in the training data. Thus instead of training a single network, the Bayesian neural network trains an ensemble of networks, where each network has its weights drawn from a shared, learnt probability distribution. 

A neural network can be viewed as a probabilistic model $P(\mathcal{Y}|\mathcal{X}, \boldsymbol {\theta})$. We assume $y$ is a continuous variable and $P(\mathcal{Y}|\mathcal{X},\boldsymbol{\theta})$ is gaussian likelihood. Give the data $Z^{M} = \{\{x_i, y_i\} \in \mathcal{X} \times \mathcal{Y}\}_{i=1}^{M}$. The total likelihood can be written as,
\begin{equation}
    p(Z|\boldsymbol{\theta}) = \prod_i\left(p(y_i|x_i, \boldsymbol{\theta})\right),
\end{equation}
which is function of network parameters $\boldsymbol{\theta}$. Assuming a gaussian likelihood, performing maximization of negative log likelihood leads to least square regression which can suffer over-fitting. Using a prior belief on the parameters $p(\boldsymbol{\theta})$ and likelihood of the data given a deep neural network model, we can obtain a updated posterior belief on the parameter using Bayes rule as follows,
\begin{equation}
    p(\boldsymbol{\theta}|Z) =  \frac{p(Z|\boldsymbol{\theta})p(\boldsymbol{\theta})}{\int_\theta p(Z|\mathbf{\theta})p(\boldsymbol{\theta})d\boldsymbol{\theta}}
    \label{eq:chap5_Bayes_rule}.
\end{equation}
Maximizing the posterior $p(\boldsymbol{\theta}|Z)$ leads to Maximum-a-posterior (MAP) estimate of parameter $\boldsymbol{\theta}$. Computing the MAP estimate has a regularizing effect and can prevent overfitting. The optimization objectives here are the same as maximizing likelihood plus a regularization term coming from the log prior. If we had a full posterior distribution over parameters we could make predictions that take weight uncertainty into account by marginalizing the parameters as follows,
\begin{equation}
    p(y|x,Z) = \int_\theta p(y|x, \boldsymbol{\theta})p(\boldsymbol{\theta}|Z)d\boldsymbol{\theta}.
\end{equation}
This is equivalent to averaging predictions from an ensemble of neural networks weighted by the posterior probabilities of their parameters $\boldsymbol{\theta}$.

\section{Adaptive Control Using Bayesian Deep Neural Networks}
\subsection{Bayesian Feature Learning: Variational Inference}
This section provides the details of the DNN training over data samples observed over n-dimensional input subspace $x(t) \in \mathcal{X} \in \mathbb{R}^{n}$ and m-dimensional target subspace $y\in\mathcal{Y} \in \mathbb{R}^m$. The sample set is denoted as $Z$ where $Z \in \mathcal{X} \times \mathcal{Y}$.

We are interested in the function approximation task for DNN. Let the function $f_{\boldsymbol{\theta}}(x)= \theta^T_n\phi_n(\theta_{n-1},\phi_{n-1}(\theta_{n-2},\phi_{n-2}(...))))$ s.t $f_{\boldsymbol{\theta}}: \mathbb{R}^n \to \mathbb{R}^m$ be the network approximating the model uncertainty with parameters $\boldsymbol{\theta} \in \boldsymbol{\Theta}$, where $\boldsymbol{\Theta}$ is the space of parameters. We assume a training data buffer $\mathcal{B}$ can store $p_{max}$ training examples, such that the set $Z^{p_{max}} = \{Z_i | Z_i \in \mathcal{Z}\}_{i=1}^{p_{max}} = \{(x_i, y_i) \in \mathcal{X}\times\mathcal{Y}\}_{i=1}^{p_{max}}$. A batch of samples can be drawn independently from the buffer $\mathcal{B}$ over probability distribution $p$ for DNN training.

Unlike the conventional DNN training where the true target values $y \in  \mathcal{Y}$ are available for every input $x \in  \mathcal{X}$, in S-DMRAC true system uncertainties as the labeled targets are not available for the network training. We use the part of the network itself (the last layer) with pointwise weight updated according to MRAC-rule as the generative model for the data. The S-DMRAC uncertainty estimates $y_i = {W}^T\Phi(x_i,\theta_1,\theta_2, \ldots \theta_{n-1})$ along with inputs $x_i$ make the training data set $Z^{p_{max}} = \{x_i, y_i\}^{p_{max}}$. We use kernel Independence test \cite{chowdhary2013concurrent} to qualify a data point to be stored in the buffer, to ensure local persistently exciting data is collected. The main purpose of DNN in the adaptive network is to extract relevant features of the system uncertainties, which otherwise is very tedious to obtain with unstructured uncertainty and unknown bounds of the domain of operation. For the details of generative network and data qualification methods refer  \cite{joshi2019deep, joshi2018adaptive} and reference therein.

Unlike Gaussian Processes \cite{liu2018gaussian}, analytical solution for posterior $p(\boldsymbol{\theta}|Z)$ for a multi-layer neural network is intractable.
We therefore approximate the true posterior with a variational distribution $q(\boldsymbol{\theta}|\boldsymbol{\zeta})$. The variational distribution is a known function parameterized by variational parameters $\boldsymbol{\zeta}$ which are estimated online. The variational parameters are estimated online by minimizing the Kullback-Leibler (KL) divergence between variational distribution $q(\boldsymbol{\theta}|\boldsymbol{\zeta})$ and true posterior $p(\boldsymbol{\theta}|Z)$. The KL objective can be written as,
\begin{eqnarray}
    KL\left(q(\boldsymbol{\theta}|\boldsymbol{\zeta})\|p(\boldsymbol{\theta}|Z)\right) &=& \int q(\boldsymbol{\theta}|\boldsymbol{\zeta}) \log\left(\frac{q(\boldsymbol{\theta}|\boldsymbol{\zeta})}{p(\boldsymbol{\theta}|Z)}\right)d\boldsymbol{\theta}, \nonumber \\
    &=&\mathbb{E}_{q(\boldsymbol{\theta}|\boldsymbol{\zeta})}\log\left(\frac{q(\boldsymbol{\theta}|\boldsymbol{\zeta})}{p(\boldsymbol{\theta}|Z)}\right).
\end{eqnarray}
The KL-divergence objective can be written as,
\begin{equation}
    KL\left(q(\boldsymbol{\theta}|\boldsymbol{\zeta})\|p(\boldsymbol{\theta}|Z)\right) = -\mathcal{L}(Z, \boldsymbol{\theta, \zeta}) +\log(p(Z)).
\end{equation}
The term $\mathcal{L}(Z, \boldsymbol{\theta, \zeta})$ is known as evidence lower bound since it defines lower bound on the evidence term $p(Z)$ as follows,
\begin{equation}
 \mathcal{L}(Z, \boldsymbol{\theta, \zeta})    = \log(p(Z))-KL\left(q(\boldsymbol{\theta}|\boldsymbol{\zeta})\|p(\boldsymbol{\theta}|Z)\right).
\end{equation}
Since the KL-divergence term is always positive semi-definite we can write,
\begin{equation}
    \mathcal{L}(Z, \boldsymbol{\theta, \zeta})    \leq \log(p(Z)).
\end{equation}
Therefore, the KL divergence between the variational distribution $q(\boldsymbol{\theta}|\boldsymbol{\zeta})$ and the true posterior $p(\boldsymbol{\theta}|Z)$ is minimized by maximizing the evidence lower bound,
\begin{equation}
    \boldsymbol{\zeta} \xleftarrow[]{}arg\max_{\boldsymbol{\zeta}}\mathcal{L}(Z, \boldsymbol{\theta, \zeta}).
\end{equation}
And the deep feature parameters are drawn from this variational distribution as follows,
\begin{equation}
    \boldsymbol{\theta} \sim q(\boldsymbol{\theta}|\boldsymbol{\zeta)}.
\end{equation}
We use Bayes by Backprop \cite{blundell2015weight} algorithm to learn the latent parameters of variational distribution using stochastic gradient descent algorithm.

\subsection{S-DMRAC Parameter Update using Bayesian Deep Features}
In \cite{joshi2019deep} we have presented the DMRAC update rule and its stability characteristics using the standard deep features. In this paper we will extend the results to stochastic deep features. The deep feature vector is defined as $\Phi^\sigma_n(x) = f_{\boldsymbol{\theta}}\backslash\theta_n$, where $f_{\boldsymbol{\theta}}$ is the entire network and $\sigma$ is the switching signal. Due to stochastic nature of the network, the feature vector $\Phi^\sigma_n(x)$ can be considered as a draw over multivariate normal distribution
\begin{equation}
    \Phi^\sigma_n(x) \sim \mathcal{N}\left(\Bar{\Phi}^\sigma_n(x), \left[G_\sigma(x)G_\sigma(x)^T\right] | \boldsymbol{\theta}\right)
    \label{eq:chap5_phi_dist}.
\end{equation}
Where $\Bar{\Phi}^\sigma_n(x), \left[G_\sigma(x)G_\sigma(x)^T\right]$ are mean and covariance  for the feature.

Using the outer layer weights $W$, updated according to MRAC rule and the deep features from Bayesian deep neural network, we can write adaptive element  $u_{ad}(t)$ in the total controller as,
\begin{equation}
    u_{ad}(t) = W^T\Phi^\sigma_n(x(t)).
    \label{eq:chap5_stochastic_adaptive_element}
\end{equation}
We will show using stochastic stability theory that the following weight update rule leads to mean square uniform ultimate boundedness almost surely. The weight update rule for the final layer weights using MRAC rule is given as
\begin{equation}
    \Dot{W}(t) = -\Gamma\Bar{\Phi}^\sigma_n(x)e(t)^TPB.
    \label{eq:S-DMRAC_parameter_update_mean}
\end{equation}
Where $\Bar{\Phi}^\sigma_n(x)$ is mean of the features given as
\begin{equation}
    \Bar{\Phi}^\sigma_n(x) = \mathbb{E}_{q(\boldsymbol{\theta}|\boldsymbol{\zeta})} \left(\Phi^\sigma_n(x)\right).
\end{equation}
However in case the true expectation is intractable, we can estimate the expectation empirically as follows,
\begin{equation}
    \hat{\Phi}^\sigma_n(x) = \frac{1}{N}\sum_{i=1}^N\left(\Phi^\sigma_{n,i}(x)\right).
    \label{eq:DMRAC_empirical_phi}
\end{equation}
Where $\Phi^\sigma_{n,i}(x)$ is the $i^{th}$ draw from distribution \eqref{eq:chap5_phi_dist} for every given $x(t)$.
Therefore we can modify the weight update law as,
\begin{equation}
    \dot{W} = proj\left(-\Gamma\left(\frac{1}{N} \sum_{i=1}^N \left(\Phi_{n,i}^\sigma(x)\right)\right) e(t)^TPB, W\right).
    \label{eq:S-DMRAC_parameter_update}
\end{equation}
However in the further section we will show that using a stochastic feature presents the property of \textit{Induced Persistency of Excitation}, which ensures the online updated weights converge to true weights neighborhood for all switching signal ``$\sigma$''.
\section{Stability Analysis for Stochastic-DMRAC}
In this section we will introduce the stochastic stability proof for S-DMRAC controller. Using the above stochastic adaptive element \eqref{eq:chap5_stochastic_adaptive_element} in the total controller \eqref{eq:total_Controller} we can rewrite the system dynamics \eqref{eq:0} as follows,
\begin{equation}
    \dot x(t) = Ax(t) + B(-k_xx(t) + k_rr(t)-u_{ad} + f(x)).
\end{equation}
Using the model matching condition we can write,
\begin{equation}
    \dot x(t) = A_{m}x(t) + B_{m}r(t)+B(f(x)-u_{ad}).
\end{equation}
Using the expression for stochastic adaptive element we can rewrite the above expression as
\begin{equation}
    \dot x(t) = A_{m}x(t) + B_{m}r(t)+B\left(f(x)-W^T\Phi_n^\sigma(x)\right).
\end{equation}
Since we know the the feature vector instance is draw over the distribution of parameters $\boldsymbol{\theta} \sim q(\boldsymbol{\theta}|\boldsymbol{\zeta}) $, we can write the random draw of the feature vector as
\begin{equation}
    \Phi_n^\sigma(x) = \Bar{\Phi}_n^\sigma + G_\sigma\xi, 
\end{equation}
where $\Bar{\Phi}_n^\sigma \in \mathbb{R}^k$ is the mean of feature distribution and $G_\sigma\in \mathbb{R}^k$ is variance and $\xi \sim \mathcal{N}(0,I)$. Using the feature instance we can write the closed loop system as
\begin{eqnarray}
\dot x(t) &=& A_{m}x(t) + B_{m}r(t) \nonumber \\
&+& B\left(f(x)-W^T\Bar{\Phi}_n^\sigma(x) - W^TG_\sigma\xi\right).
\end{eqnarray}
We assume for every $\Bar{\Phi}_n^\sigma$ that there exists an ideal weight $W^*$, such that,
\begin{equation}
    f(x) = W^{*T}\Bar{\Phi}_n^\sigma(x) + \epsilon_m^\sigma(x).
\end{equation}
Where $\epsilon_m^\sigma(x) = \Delta(x) - W^{*T}\Bar{\Phi}_n^\sigma(x)$ is network approximation error due to universal approximation theorem, and this error can be bounded as, 
\begin{equation}
    \Bar{\epsilon}^\sigma = \sup_{x \in \mathcal{D}}\left\|f(x) - W^{*T}\Bar{\Phi}_n^\sigma(x)\right\|.
\end{equation}
Using \eqref{eq:ref model}, we can write the error dynamics as,
\begin{equation}
    \dot e(t) = A_{m}e(t)+B\left(\Tilde{W}^T\Bar{\Phi}_n^\sigma(x)+\epsilon_m^\sigma(x)\right) + BW^TG_\sigma\xi.
\end{equation}
Where $\Tilde{W} = W^*-W$.
We can write the above tracking error dynamics as diffusion process as follows,
\begin{eqnarray}
de(t) &=& A_{rm}e(t)dt+B\left(\Tilde{W}^T\Bar{\Phi}_n^\sigma(x)+\epsilon_m^\sigma(x)\right)dt \nonumber \\ &+& BW^TG_\sigma d\xi(t),
    \label{eq:chap5_error_dynamics}
\end{eqnarray}

where $d\xi(t)$ is zero mean Wiener process.
\begin{theorem}
Consider the system in \eqref{eq:0}, the control law of \eqref{eq:total_Controller}, and assume that the uncertainty $f(x):\mathbb{R}^n \to \mathbb{R}^m$ is bounded and Lipschitz continuous on a compact set $x(t) \in \mathcal{D}_x$. Let  the reference signal $r(t)$ is such that the state $x_{m}(t)$ of the bounded input bounded output reference model \eqref{eq:ref model} remains
bounded in the compact ball $\mathcal{B}_m = \{x_{m} : \|x_{m}(t)\| \leq \mathcal{B_\delta}\}$, then the adaptive control $u_{ad} = W^T\Phi^\sigma_n(x)$, guarantee that
the system is mean square uniformly ultimately bounded in
probability a.s.
\end{theorem}
\begin{proof}
Let,
\begin{equation}
    V(e_t, \Tilde{W}(t)) = \frac{1}{2}e_t^TPe_t + \frac{1}{2}tr\left(\Tilde{W}(t)^T\Gamma^{-1}\Tilde{W}(t)\right),
\end{equation}
be stochastic lyapunov candidate. The error term $e_t$ is a realization of drift diffusion process, we assume $t$ as index and not function of time, whereas $W(t)$ is a deterministic process $W(t)$ is a function of time. Note that the lyapunov candidate is bounded above below by

\begin{equation}
\underline{\zeta}\left(e_t, \Tilde{W}\right) \leq V(e_t, \Tilde{W}) \leq \bar{\zeta}\left(e_t, \Tilde{W}\right),
\end{equation}
where $\underline{\zeta}\left(e_t, \Tilde{W}\right) = \frac{1}{2}\underline{\lambda}(P)\|e_t\|^2 + \frac{1}{2}\underline{\lambda}\left(\Gamma^{-1}\right)\|\Tilde{W}\|_F^2$ and $\bar{\zeta}\left(e_t, \Tilde{W}\right) = \frac{1}{2}\bar{\lambda}(P)\|e_t\|^2 + \frac{1}{2}\bar{\lambda}\left(\Gamma^{-1}\right)\|\Tilde{W}\|_F^2$ and $\underline{\lambda}$ and $\bar{\lambda}$ is the minimum and maximum Eigen value operator and $\|\|_F$ is the Frobenius norm. 

The Ito derivative of the Lyapunov candidate along the error dynamics \eqref{eq:chap5_error_dynamics} for $\sigma^{th}$ system,
\begin{eqnarray}
&&\hspace{-10mm}\mathcal{L}V\left(e_t, \tilde{W}\right) = - \dot{W}\Gamma^{-1}\tilde{W}-\frac{1}{2}e^TQe + \tilde{W}\Bar{\phi}_n^\sigma(x)e^TPB \nonumber\\
&&+e^TPB \epsilon_m^{\sigma}(z) + \frac{1}{2}[(BW^TG_{\sigma})(BW^TG_{\sigma})^TP].
\end{eqnarray}
Using the weight update law \eqref{eq:S-DMRAC_parameter_update_mean}, we can simplify the above expression as follows
\begin{eqnarray}
\hspace{-10mm}\mathcal{L}V\left(e_t, \tilde{W}\right) &=& \frac{1}{2}e^TQe + e^TPB\epsilon_m^{\sigma}(z) \nonumber \\
&+& \frac{1}{2}[(BW^TG_{\sigma})(BW^TG_{\sigma})^TP].
\end{eqnarray}
Let $c_1 = 1/2\|P\|\|W\|_F^2\|BG_{\sigma}\|^2$ and $c_2 = \|PB\|$, then
\begin{eqnarray}
\mathcal{L}V\left(e_t, \tilde{W}\right) 
&\leq& -\frac{1}{2} \lambda_{min}(Q)\|e\|^2 + c_2\|e\|\| \epsilon_m^{\sigma}(z)\| + c_1 \nonumber
\end{eqnarray}
We know $\sup_{x\in\mathcal{D}}\|\epsilon_m^{\sigma}(x)\| \leq \bar{\epsilon}^\sigma$, and using the projection operator the outer layer weights $W(t)$ are bounded $\|W\| \leq \mathcal{W}_b$.
Using this bound in above expression we can write,
\begin{equation}
\mathcal{L}V\left(e_t, \tilde{W}\right)  \leq -\frac{1}{2} \lambda_{min}(Q)\|e\|^2 + c_2\|e\| \bar{\epsilon}^\sigma + c_1.
\end{equation}
Let $c_5^{\sigma} = c_2\bar{\epsilon}^\sigma$.

We can define an set outside which the $\mathcal{L}V\left(e_t, \tilde{W}\right)\leq 0$
\begin{equation}
\Theta^{\sigma} = \left\lbrace\|e\| \geq \frac{c_5^{\sigma} + \sqrt{\left(c_5^{\sigma}\right)^2 + 2*\lambda_{min}(Q)c_1}}{\lambda_{min}(Q)} \right\rbrace.
\end{equation}
Therefore, outside the set $\Theta^{\sigma}$,  $\mathcal{L}V\left(e_t, \tilde{W}\right)  \leq 0$ a.s. Using BIBO property of reference model \eqref{eq:ref model}, when $r(t)$ is bounded, $x_{m}(t)$ remains bounded
within $B_m$. Solution to \eqref{eq:0} $x(t) \in \mathcal{D}_x$ a.s. Since this is true for all $\sigma$, and because Algorithm
guarantees that $\sigma$ does not switch arbitrarily fast,\eqref{eq:chap5_error_dynamics} is mean square uniformly ultimately bounded inside set $\Theta^\sigma$ a.s.
\end{proof}

\section{Persistency of Excitation for S-DMRAC}
In multi-layer architecture for uncertainty approximation, it is challenging to ensure the Persistency of Excitation, which result in non P.E deep features and therefore non convergence of the outer-layer weights.  One straightforward way to remedy
this problem is to inject exogenous perturbations into the dynamics of the gradient descent algorithm so that even the parameters in the hidden layers receive persistent excitation during training. The traditional procedure to introduce robustness for linear models is the addition of regularization term into the loss function. However the regularization cannot be re-framed as a method to ensure persistent excitation but can bound model parameters. 

However, we will show that Bayesian-deep neural network is a natural way of adding regularization and the stochastic nature of network weights add randomness in the network prediction and therefore lead to desirable property of induced persistency of excitation.

The KL-objective for posterior inference given data in training Bayesian-deep neural network is follows,
\begin{equation}
    \mathcal{L}(Z, \boldsymbol{\theta, \zeta}) = \mathbb{E}_{q(\boldsymbol{\theta}|\boldsymbol{\zeta})}\left(\log\left(p(Z|\boldsymbol{\theta})\right)\right) - KL\left(q(\boldsymbol{\theta}|\boldsymbol{\zeta})\|p(\boldsymbol{\theta})\right).
\end{equation}
While maximizing $\mathcal{L}(Z, \boldsymbol{\theta, \zeta})$, we are maximizing the model likelihood given data $Z$ and minimizing KL-divergence between posterior and prior on parameters, which is akin to adding regularization.

Uncertainty in predictions arise from two sources; uncertainty in weights called Epistemic uncertainty and the variability coming from the finite excitation in training data know as Aleatoric uncertainty. Epsitemic uncertainty can grows smaller, if we have more training data. Consequently, epistemic uncertainty is higher in regions of no or little training data and lower in regions of more training data. Epistemic uncertainty is covered by the variational posterior distribution.  Aleatoric uncertainty is covered by the probability distribution used to define the likelihood function. The Aleatoric uncertainty cannot be reduced if we get more data, it is the inherent variability in the training data.
The total uncertainty or randomness in the feature vector due to Epistemic and Aleatoric uncertainty, lead to favorable property of induced persistency of excitation of $\Phi^\sigma_n$

\begin{theorem}
\label{theorem:DMRAC_PE_proof}
Let the Bayesian deep network $f_{\boldsymbol{\theta}}$  be trained over data $Z^{p_{max}} = \{\{x_i, y_i\} \in \mathcal{X} \times \mathcal{Y}\}_{i=1}^{p_{max}}$, such that data is collected using kernel independence test. Then the features $\Phi^\sigma_n = f_{\boldsymbol{\theta}}\backslash\theta_n$ is persistently exciting.
\end{theorem}
\begin{proof}
We use autocovariance argument to prove the that the stochastic features $\Phi^\sigma_n$ are persistently exciting i.e iff
\begin{equation}
    R_{\Phi^\sigma_n}(0) > 0,
\end{equation}
Implies the integral for any given $t \geq t_0$ and over any interval $[t, t+T]$ for any $T > 0$ satisfies,
\begin{equation}
    \int_t^{t+T}\Phi^\sigma_n(\tau)\Phi^\sigma_n(\tau)^Td\tau \geq \gamma I.
\end{equation}
To prove the fact that deep features of S-DMRAC are PE, we show that the autocovariance $R_{\Phi^\sigma_n}(0)$ exists and $R_{\Phi^\sigma_n}(0) > 0$.

Since the stochastic feature vector are realization of distribution defined as
\begin{equation}
    \Phi^\sigma_n(x) \sim \mathcal{N}\left(\Bar{\Phi}^\sigma_n(x), \left[G_\sigma(x)G_\sigma(x)^T\right] | \boldsymbol{\theta}\right).
    \label{eq:chap5_phi_dist}
\end{equation}
Therefore the autocovariance of feature vector $\Phi^\sigma_n$ can be defined as,
\begin{eqnarray}
    &&\hspace{-15mm}R_{q_{\zeta}(\Phi^\sigma_n|x,\theta)} \left(\Phi^\sigma_n(x)|x\right) =\mathbb{E}_{q_\theta}\left[\left(\Phi^\sigma_n-\mathbb{E}(\Phi^\sigma_n)\right)^2\right] \nonumber \\
    &=&\mathbb{E}_{q_\theta}\left(\Phi^\sigma_n\Phi^{\sigma_n^{T}}\right)-\mathbb{E}_{q_\theta}\left(\Phi^\sigma_n\right)\mathbb{E}_{q_\theta}\left(\Phi^\sigma_n\right)^T,
    \label{eq:chap5_autocovar}
\end{eqnarray}
Where $q_{\boldsymbol{\zeta}}(\Phi^\sigma_n|x,\theta)$ is the variational predictive posterior distribution defined as,
\begin{eqnarray}
q_{\boldsymbol{\zeta}}(\Phi^\sigma_n|x,\theta) &=& p(\Phi^\sigma_n|x, \theta)q(\boldsymbol{\theta}|\boldsymbol{\zeta}).
\end{eqnarray}
and $p(\Phi^\sigma_n|x, \theta)$ is the posterior distribution on feature given a input $x$ and network parameters $\boldsymbol{\theta}$
\begin{equation}
    p(\Phi^\sigma_n|x, \theta) = \mathcal{N}\left(\Bar{\Phi}^\sigma_n(x), \left[G_\sigma(x)G_\sigma(x)^T\right] | \boldsymbol{\theta}\right).
\end{equation}
Using the above definitions of conditional distributions, the autocovariance term \eqref{eq:chap5_autocovar} can be decomposed into \textbf{\textit{Aleatoric}} and \textbf{\textit{Epistemic}} uncertainties as follows \cite{kwon2018uncertainty}:
\begin{eqnarray}
&&\hspace{-15mm}R_{q}\left(\Phi^\sigma_n(x)|x\right) \mathbb{E}_q\left(\Phi^\sigma_n\Phi^{\sigma_n^{T}}\right)-\mathbb{E}_q\left(\Phi^\sigma_n\right)\mathbb{E}_q\left(\Phi^\sigma_n\right)^T \nonumber \\
&&\hspace{-10mm}\Rightarrow \int_\Theta diag\left(\mathbb{E}_{p(\Phi^\sigma_n|x, \theta)}\left(\Phi^\sigma_n\right)\right) \nonumber \\
&& -\mathbb{E}_{p(\Phi^\sigma_n|x, \theta)}\left(\Phi^\sigma_n\right)\mathbb{E}_{p(\Phi^\sigma_n|x, \theta)}\left(\Phi^\sigma_n\right)^Tq_{\boldsymbol{\zeta}}(\boldsymbol{\theta})d\boldsymbol{\theta}\nonumber \\
&&\hspace{-10mm}+\int_\Theta\left[\mathbb{E}_{p(\Phi^\sigma_n|x, \theta)}\left(\Phi^\sigma_n\right)-\mathbb{E}_{q_{\zeta}(\Phi^\sigma_n)}\left(\Phi^\sigma_n\right)\right] \times \nonumber \\
&&\left[\mathbb{E}_{p(\Phi^\sigma_n|x, \theta)}\left(\Phi^\sigma_n\right)-\mathbb{E}_{q_{\zeta}(\Phi^\sigma_n})\left(\Phi^\sigma_n\right)\right]^T q_{\boldsymbol{\zeta}}(\boldsymbol{\theta})d\boldsymbol{\theta}.
\nonumber\\
\label{eq:chap5_total_variance}
\end{eqnarray}
The first term of the predictive variance of the variational posterior distribution \eqref{eq:chap5_total_variance} in known as Aleatoric uncertainty. The expectation is with respect to predictive posterior distribution $p(\Phi^\sigma_n|x, \theta)$ integrated over all parameters $\boldsymbol{\theta}$. The second term of the predictive variance of the variational posterior distribution \eqref{eq:chap5_total_variance} is the epistemic uncertainty, variance in predictions due to uncertainty in weights.
\begin{figure}[tbh!]
    \centering
    \includegraphics[width=0.75\columnwidth]{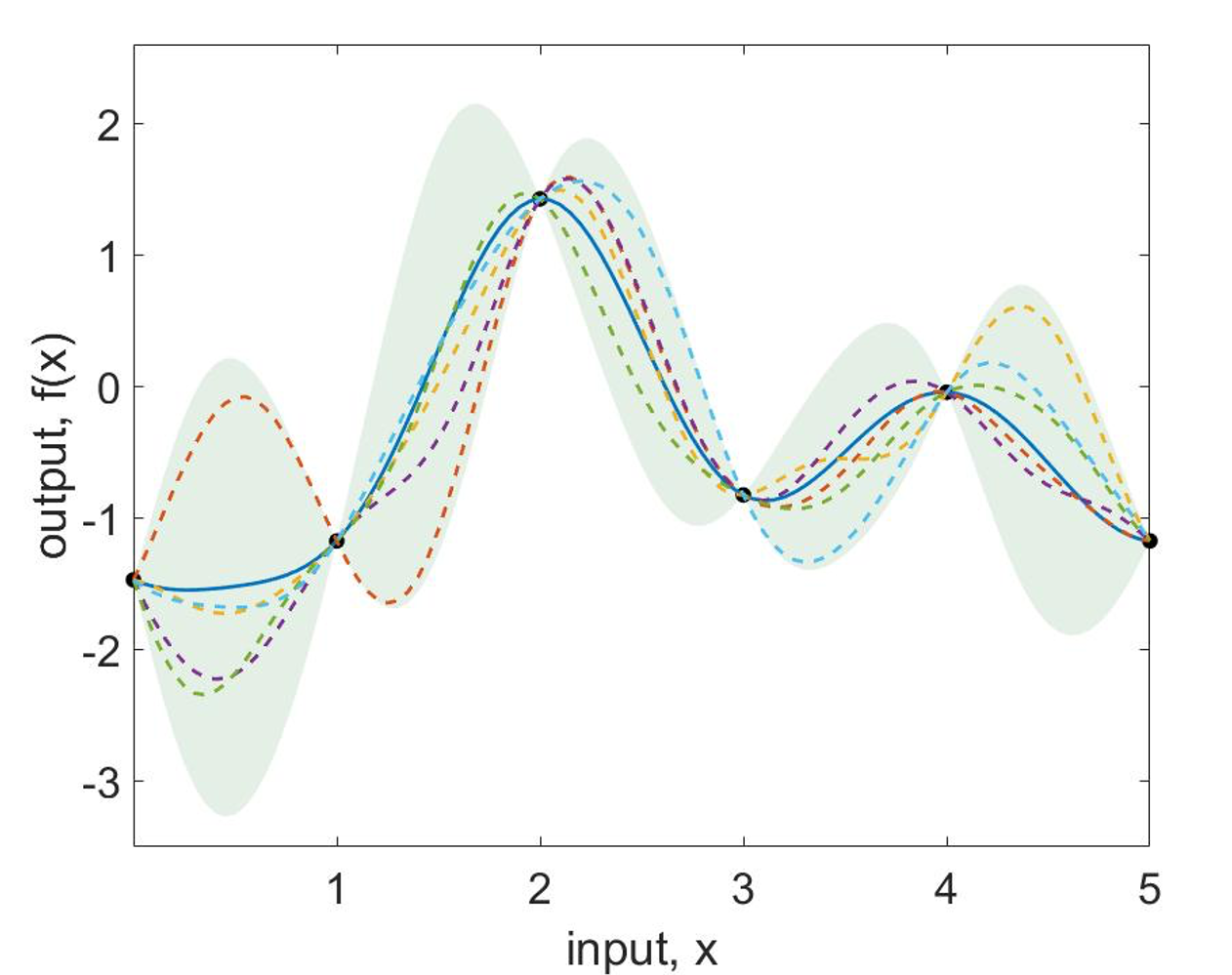}
    \caption{Bayesian Neural Network: Shaded area is non zero Epistemic(prediction) uncertainty and Black dots are training points chosen to ensure non zero Aleatoric uncertainty (kernel independence test)}
    \label{fig:chap5_GP_inference}
\end{figure}

If we chose the training data such that Aleatoric uncertainty is non zero, then we can show that the autocovariance of features $\Phi^\sigma_n$ is non zero everywhere. We take benefit of finite excitation property, and using kernel independence test \cite{chowdhary2013concurrent} and qualify the training data points, such that are sufficiently distinct from the stored data in buffer $\{Z: x, y \in \mathcal{X}, \mathcal{Y}\}_{i=1}^{p_{max}}$ hence ensuring Aletoric uncertainty is non zero and due to fact that the network weights are stochastic the epistemic/prediction uncertainty is non zero everywhere else.

Hence using stochastic Bayesian network and training the network over data points chosen according to kernel independence test we can ensure the total variance $R_{q}\left(\Phi^\sigma_n(x)|x\right) > 0$ and therefore $\Phi^\sigma_n$ is P.E
\end{proof}
\subsection{Uniform boundedness of the Outer-Layer parameters}
Consider the problem of parameter estimation using a linear estimator with stochastic deep features. We learn a given function by minimizing its squared loss error.

Let us assume the function to be estimated be $f: \mathbb{R}^n \to \mathbb{R}^m$. Let the true function $f$ be of form 
\begin{equation}
    f(x) = W^{*T}\Bar{\Phi}_n^\sigma(x) + \epsilon(x),
    \label{eq:DMRAC_ls_system}
\end{equation}
Where the unknown function be linearly parameterized in unknown ideal weight $W^* \in \mathbb{R}^{k \times m}, \forall \sigma$ and $\bar{\Phi}_n^\sigma(x): \mathbb{R}^n \to \mathbb{R}^k$ such that
$
\bar{\Phi}_n^\sigma(x) = \mathop{\mathbb{E}}_\theta\left({\Phi}_n^\sigma(x)\right)
$.
Where ${\Phi}_n^\sigma(x)$ is stochastic nonlinear continuously differentiable basis function which is drawn from distribution \eqref{eq:chap5_phi_dist}

Let the $W(t)$ be the online estimate of true ideal weights $W^*$, therefore the online estimate of $f(x)$ can be represented by mapping $\nu: \mathbb{R}^n \to \mathbb{R}^m$ such that,
\begin{eqnarray}
    \nu(x) &=& W^T{\Phi}_n^\sigma(x)= W^T\left(\bar{\Phi}_n^\sigma(x)+G_\sigma\xi\right).
    \label{eq:DMRAC_ls_estimate_model}
\end{eqnarray}
Let $e(t)$ be defined as,
\begin{eqnarray}
    e(t) &=& f(x) - \nu(x) \nonumber \\
    &=& W^{*T}\Bar{\Phi}_n^\sigma(x) + \epsilon(x) - W^T\left(\bar{\Phi}_n^\sigma(x)+G_\sigma\xi\right) \nonumber \\
    &=&\Tilde{W}^T\Bar{\Phi}_n^\sigma(x) + \epsilon(x)-W^T\left(G_\sigma\xi\right).
    \label{eq:ls_error}
\end{eqnarray}
the parameter update using gradient descent algorithm is given as
\begin{equation}
    \dot{W} = -\Gamma\bar{\Phi}_n^\sigma(x)e(t)
    \label{eq:DMARC_Stochastic_param_update}
\end{equation}
The parameter error dynamics can be found by differentiating $\Tilde{W}$ and using equation \eqref{eq:DMARC_Stochastic_param_update} as,
\begin{equation}
    \dot{\tilde{W}} = -\Gamma\bar{\Phi}_n^\sigma(x)e(t),
    \label{eq:DMRAC_ls_weight_update}
\end{equation}
    
This is a linear time varying differential equation in $\Tilde{W}$. Furthermore, note that if PE Condition using Theorem-\ref{theorem:DMRAC_PE_proof} is satisfied, then $\Tilde{W}$ is asymptotically bounded near zero solution.

\begin{theorem}
\label{theorem:DMRAC_weight_conv}
Consider the system model given by equation \eqref{eq:DMRAC_ls_system}, the online estimation model given by equation \eqref{eq:DMRAC_ls_estimate_model}, the stochastic deep feature gradient descent weight update law \eqref{eq:DMRAC_ls_weight_update}, and assume that the regressor function ${\Phi}_n^\sigma(x)$ is continuously differentiable and that the measurements ${\Phi}_n^\sigma(x) \in \mathcal{D}$ where $\mathcal{D}\in \mathbb{R}^m$ is a
compact set. If the ${\Phi}_n^\sigma(x)$ satisfy the PE condition, then the zero solution
of the weight error dynamics of equation \eqref{eq:DMRAC_ls_estimate_model} is globally uniformly ultimately bounded.
\end{theorem}
\begin{proof}
Consider a quadratic Lyapunov function 
\begin{equation}
    V(\tilde{W}) = \frac{1}{2}tr\left(\Tilde{W}^T\Gamma^{-1}\Tilde{W}\right).
    \label{eq:ls_lyapunov}
\end{equation}
Such that  $V(0) = 0$ and  $V(\tilde{W}) > 0, \forall \tilde{W} \neq 0$.  Since $V(\tilde{W})$  is quadratic, letting $\lambda_{min}(.)$ and $\lambda_{max}(.)$ denote the operators that return the minimum and maximum eigenvalue of a matrix, we have:
$$
\lambda_{min}\left(\Gamma^{-1}\right)\left\|\Tilde{W}\right\|_F^2 \leq V(\tilde{W}) \leq \lambda_{max}\left(\Gamma^{-1}\right)\left\|\Tilde{W}\right\|_F^2.
$$
Taking a time derivative along the trajectory \eqref{eq:DMRAC_ls_weight_update}, we have
\begin{eqnarray}
\dot{V}(\tilde{W}) &=& tr\left(\Tilde{W}^T\Gamma^{-1}\dot{\Tilde{W}}\right).
\end{eqnarray}
Using \eqref{eq:DMRAC_ls_weight_update}, \eqref{eq:ls_error} and the definition of empirical estimate of $\bar{\Phi}_n^\sigma(x)$ in \eqref{eq:DMRAC_empirical_phi}, the above expression cn be rewritten as,
\begin{eqnarray}
\dot{V}(\tilde{W}) &=& -\Tilde{W}^T\bar{\Phi}_n^\sigma(x)\bar{\Phi}_n^\sigma(x)^T\Tilde{W}-\bar{\Phi}_n^\sigma(x)\Tilde{W}^T\epsilon(x) \nonumber \\
&&-\Tilde{W}^T\bar{\Phi}_n^\sigma(x)\left(G_\sigma\xi\right)^T\Tilde{W}\\
&\leq& -\lambda_{max}(\Omega)\left\|\Tilde{W}\right\|^2-\left\|\bar{\Phi}_n^\sigma\right\|_\infty\mathcal{W}_b\bar{\epsilon}^\sigma ,\\
&\leq& -\frac{\lambda_{max}(\Omega)}{\lambda_{min}\left(\Gamma^{-1}\right)}V(\tilde{W})-\left\|\bar{\Phi}_n^\sigma\right\|_\infty\mathcal{W}_b\bar{\epsilon}^\sigma.
\end{eqnarray}
Where $\Omega$ is defined as follows,
\begin{equation}
    \Omega = \frac{1}{N}\sum_{i,j = 1}^N\left(\Phi_{n,i}^\sigma(x)\bar{\Phi}_{n,j}^\sigma(x)^T + \Phi_{n,i}^\sigma(x)(G_\sigma\xi_j)^T\right).
\end{equation}
Using P.E condition of $\Phi_{n,i}^\sigma(x)$, we know that $\lambda_{max}(\Omega) > 0$.
Therefore using Lyapunov stability theory if
\begin{equation}
    \left\|\Tilde{W}\right\| \geq \sqrt{\frac{\left\|\bar{\Phi}_n^\sigma\right\|_\infty\mathcal{W}_b\bar{\epsilon}^\sigma}{\lambda_{max}(\Omega)}},
\end{equation}
we have $\dot{V}(\tilde{W}) \leq 0$. Therefore
the set $W_\delta = \left\{\left\|\Tilde{W}\right\| \leq \sqrt{\frac{\left\|\bar{\Phi}_n^\sigma\right\|_\infty\mathcal{W}_b\bar{\epsilon}^\sigma}{\lambda_{max}(\Omega)}}\right\}$ is positively invariant, hence $\tilde{W}$ approaches the bounded near zero solution exponentially fast and remain are ultimately bounded.
\end{proof}
If Theorem-\ref{theorem:DMRAC_weight_conv} holds, then the adaptive weights $W(t)$ will approach exponentially fast and remain bounded in a compact neighborhood of the ideal weights $W^*$.

\section{S-DMRAC evaluation: Wing-Rock System}
This section evaluates the S-DMRAC controller on the Wing-Rock system \cite{luo1993control, monahemi1996control}. The aim of the controller is to follow the step reference command under unknown nonlinear uncertainties.

Let $\phi$ denote the roll angle of an
aircraft, $p$ denote the roll rate, $\delta_a$ denote the aileron control input, then a simplified model for wing rock dynamics is given by
\begin{eqnarray}
    \Dot{\phi} &=& p, \hspace{4mm}\\
    \Dot{p} &=& \delta_a + \Delta(x).
\end{eqnarray}
The system uncertainty is defined as
\begin{equation}
    \Delta(x) = W_0 + W_1\phi + W_2p + W_3|\phi|p + W_4|p|p + W_5\phi^3
\end{equation}
Where $x = [\phi, p]$. The true parameters are as follows
$[W_0 = 1.0,W_1 = 0.2314,W_2 =0.6918,W_3 = −0.6245,W_4 = 0.1,W_5 = 0.214]$.
The task of the controller is to follow a reference step command $r(t)$. The reference model chosen is a stable second order linear system with natural frequency of $2rad/s$ and damping ratio of $0.5$. The linear control gains are given by $k = [-4, -2]$ and $k_r = -4$ and the learning rate is set to $\Gamma = 10$. Initial conditions for the simulation are arbitrarily chosen to be $\phi = 1deg$, $p = 1deg/s$. The tolerance threshold for kernel independence test is selected to be $\epsilon_{tol} = 0.1$ for updating the basis vector $\mathcal{BV}(\sigma)$. Using high bias ''$W_0$", varied step size in tracking command and higer learning rate, we push the controller to their limits, to test the robusteness of S-DMRAC vs DMRAC controller.

Figure-\ref{fig:SDMRAC_states} show the closed-loop system performance in tracking the reference signal for the S-DMRAC controller. As we can see, the DMRAC has a larger oscillation in state tracking compared to S-DMRAC.
The presented controller S-DMRAC achieves better uncertainty estimate as shown in Fig-\ref{fig:SDMRAC_nuad} compared to the DMRAC controller. Due to the Bayesian property of including prior information to parameter update, which introduces natural regularization, for the same learning rate, S-DMRAC estimates are less noisy compared to DMRAC. Due to the stochastic nature of the S-DMRAC controller,  we also derive the confidence bound around the estimates Fig-\ref{fig:SDMRAC_nuad}. We can Observe around the region of the plot where the tracking signal has a variation (finitely exciting), the prediction confidence is high. The prediction confidence is low around the non P.E region of the reference signal; therefore, a large variance  Fig-\ref{fig:SDMRAC_nuad}. Due to this overall variation (Aleatoric and Epistemic), the deep features are characterized by the induced persistency of excitation, resulting in network weights $W$ for S-DMRAC evolve steady and display convergence-like properties compared to DMRAC Fig-\ref{fig:SDMRAC_Ws}.
\begin{figure}[tbh!]
    \includegraphics[width=1.0\linewidth]{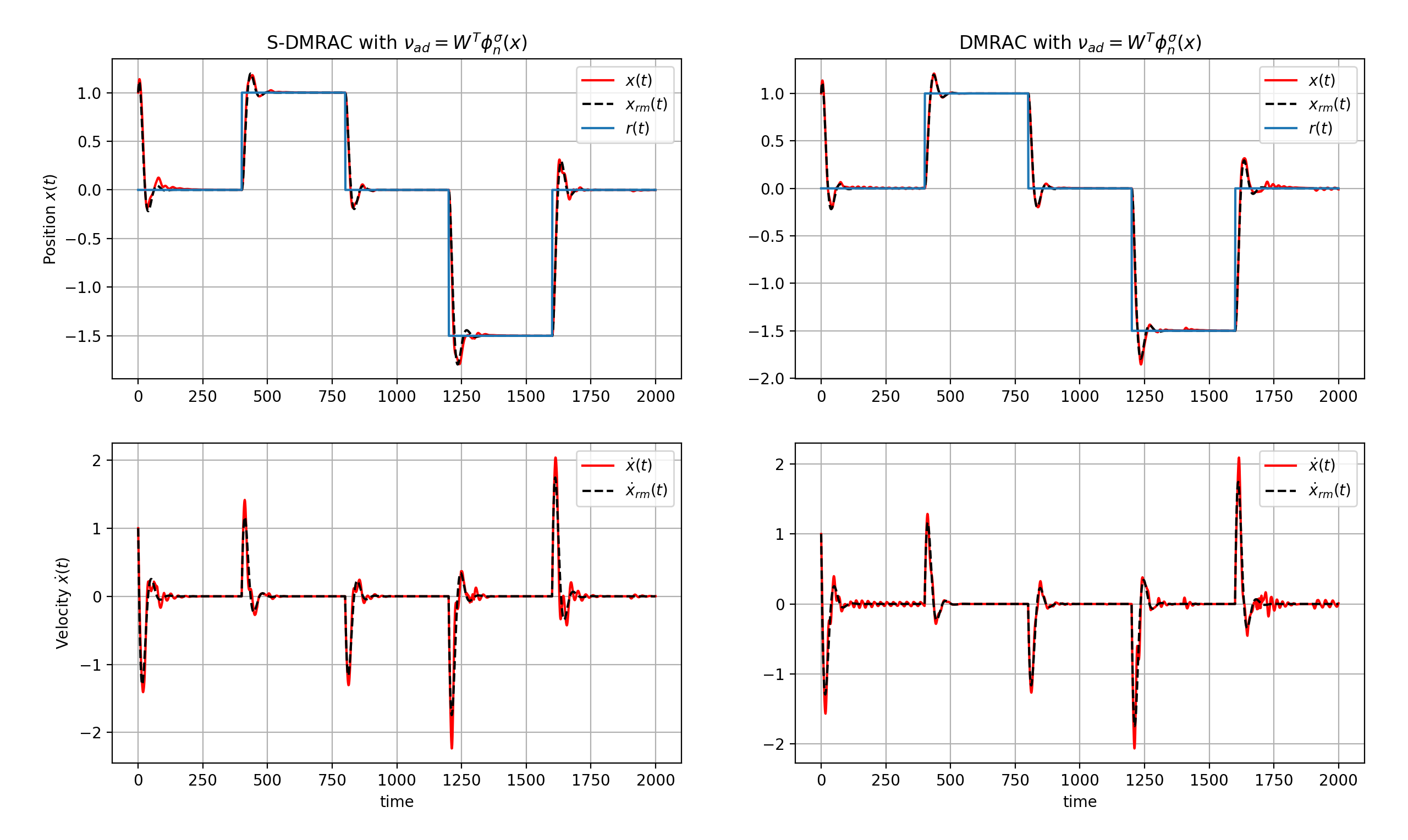}
    \caption{Wind-Rock Roll and Roll-rate for step command following using S-DMRAC controller}
    \label{fig:SDMRAC_states}
\end{figure}
\begin{figure}[tbh!]
\vspace{-5mm}
    \centering
    \includegraphics[width=1\linewidth]{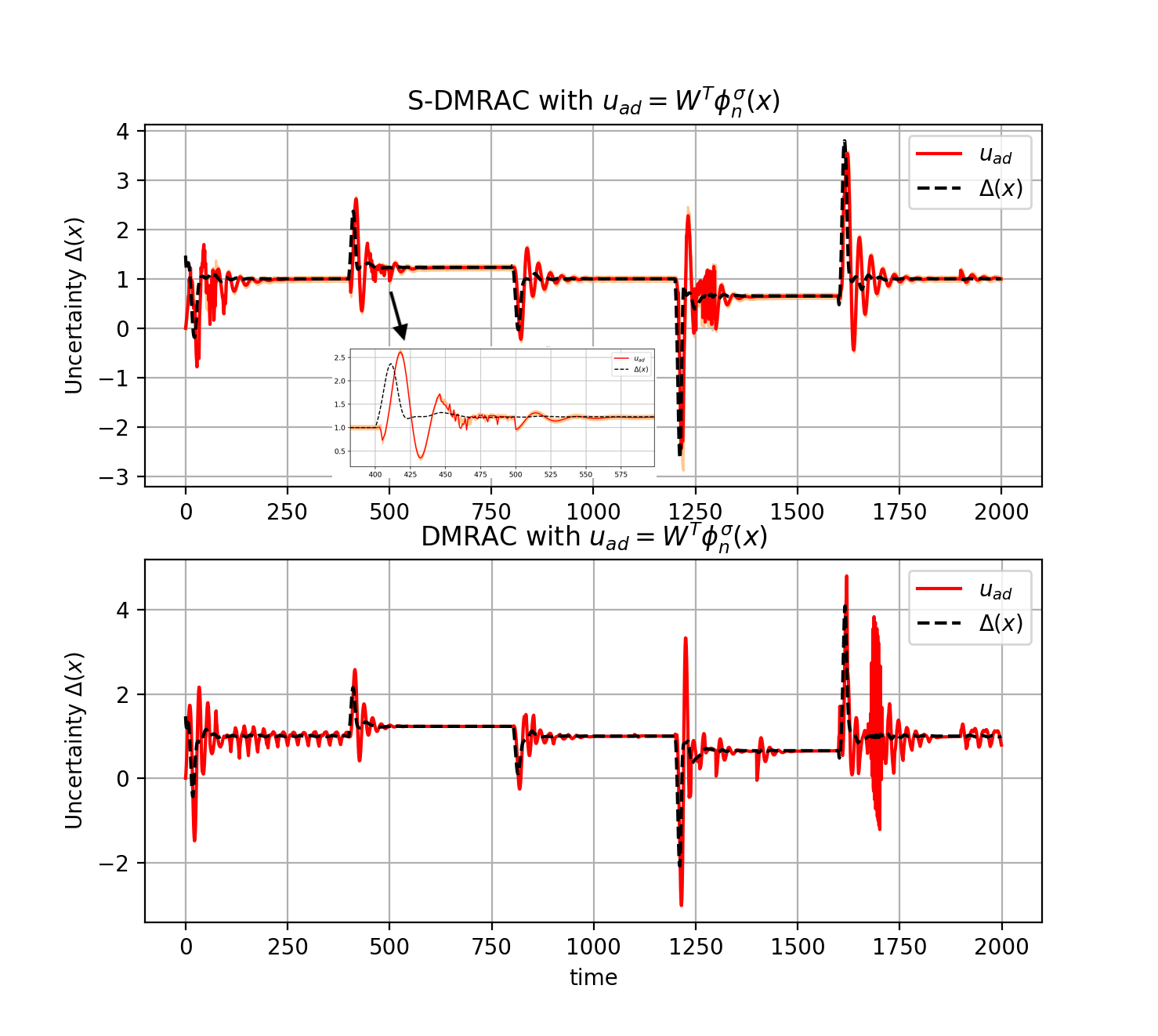}
    \caption{Uncertainty Estimate S-DMRAC vs DMRAC}
    \label{fig:SDMRAC_nuad}
\end{figure}
\begin{figure}[tbh!]
    \vspace{-4mm}
    \centering
    \includegraphics[width=1\linewidth]{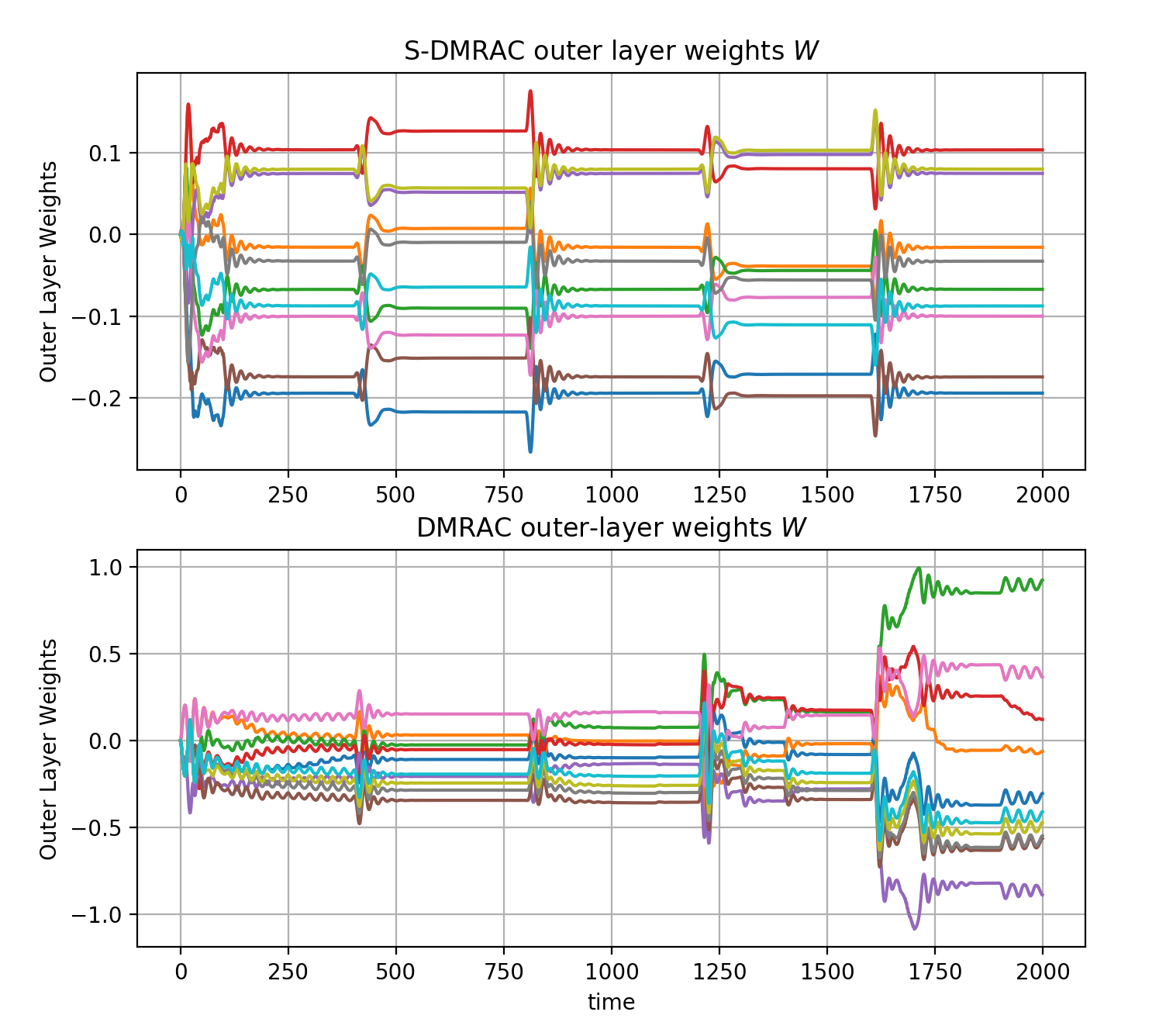}
    \caption{Outer-Layer weights $W$ evolution with time for S-DMRAC vs DMRAC}
    \label{fig:SDMRAC_Ws}
    \vspace{-3mm}
\end{figure}

\section{Conclusion}
\label{conclusions}
This paper presented an S-DMRAC adaptive controller using Bayesian neural networks to address feature designing for unstructured uncertainty. The proposed controller uses Bayesian DNN to model significant uncertainties without knowledge of the system's domain of operation. We provide theoretical proofs that, given the network predictions' stochastic nature, ensures "\textit{Induced persistency of excitation}" of the features leading to convergence of both state and network final layer weights. Numerical simulations with the Wing-rock model demonstrate the controller performance in achieving reference tracking in the presence of significant matched uncertainties. We compared the results with a state-of-the-art DMRAC controller and demonstrated the network weights have less transient and show convergence. We conclude that both DMRAC and S-DMRAC are highly efficient architectures for high-performance control of nonlinear systems with robustness and long-term learning properties.

\bibliographystyle{unsrt}
\bibliography{main}

\end{document}